\documentclass[12pt]{article}

\usepackage{latexsym, amssymb, amscd, amsthm, amsxtra, amsmath,amsthm }
\usepackage{graphics, graphicx, color}
\usepackage{natbib}
\usepackage{hyperref}
\usepackage{ifpdf}
\usepackage[format=hang,indention=-1cm,small]{caption}
\usepackage[caption=false]{subfig}
\usepackage{multirow}
\usepackage[final]{pdfpages}
\graphicspath{{./images/}}

\newtheorem{thm}{Theorem}

\newtheorem{lem}[thm]{Lemma}
\newtheorem{prop}[thm]{Proposition}
\theoremstyle{definition}

\newtheorem*{defn*}{Definition}
\theoremstyle{remark}
\newtheorem{remark}{Remark}

\newcommand{\norm}[1]{\left\Vert#1\right\Vert}

\newcommand{\Real}{\mathbb R}

\def\argmin{\mathop{\rm argmin}}

\def\Cov{\mbox{Cov}}

\def\E{\mbox{E}}
\def\tr{\mbox{trace}}
\def\half{\frac{1}{2}}

\def\rank{\mbox{rank}}

\def\smt{\text{\tiny T}}
\def\av{\mathbf a}
\def\bv{\mathbf b}

\def\gv{\mathbf g}
\def\hv{\mathbf h}

\def\gv{\mathbf g}

\def\qv{\mathbf q}

\def\uv{\mathbf u}
\def\vv{\mathbf v}
\def\wv{\mathbf w}
\def\xv{\mathbf x}
\def\yv{\mathbf y}
\def\zv{\mathbf z}

\def\Av{\mathbf A}
\def\Bv{\mathbf B}
\def\Cv{\mathbf C}
\def\Dv{\mathbf D}

\def\Iv{\mathbf I}
\def\Jv{\mathbf J}

\def\Lv{\mathbf L}
\def\Mv{\mathbf M}
\def\Nv{\mathbf N}

\def\Qv{\mathbf Q}
\def\Rv{\mathbf R}
\def\Sv{\mathbf S}
\def\Tv{\mathbf T}
\def\Uv{\mathbf U}
\def\Vv{\mathbf V}
\def\Wv{\mathbf W}
\def\Xv{\mathbf X}

\def\Zv{\mathbf Z}

\def\z{{\sf z}}
\def\q{{\sf q}}

\newcommand{\alphav}{\mbox{\boldmath{$\alpha$}}}
\newcommand{\betav}{\mbox{\boldmath{$\beta$}}}

\newcommand{\deltav}{\mbox{\boldmath{$\delta$}}}

\newcommand{\muv}{\mbox{\boldmath{$\mu$}}}

\newcommand{\Sigmav}{\mbox{\boldmath{$\Sigma$}}}
\newcommand{\Lambdav}{\mbox{\boldmath{$\Lambda$}}}

\newcommand{\Cc}{\mathcal{C}}

\newcommand{\Oc}{\mathcal{O}}

\newcommand{\Uc}{\mathcal{U}}

\def\1v{\mathbf 1}
\def\0v{\mathbf 0}
\def\Id{ \Iv}

\newcommand{\blind}{0}

\begin{document}

\def\spacingset#1{\renewcommand{\baselinestretch}%
{#1}\small\normalsize} \spacingset{1}


\if0\blind
{
  \title{\bf Penalized Orthogonal Iteration for Sparse Estimation of Generalized Eigenvalue Problem}
  \author{Sungkyu Jung\\
    Department of Statistics, University of Pittsburgh\\
        Jeongyoun Ahn \\
    Department of Statistics, University of Georgia\\
    and \\
    Yongho Jeon\\
    Department of Applied Statistics, Yonsei University}
  \maketitle
} \fi

\if1\blind
{
  \bigskip
  \bigskip
  \bigskip
  \begin{center}
    {\LARGE\bf Penalized Orthogonal Iteration for Sparse Estimation of Generalized Eigenvalue Problem}
\end{center}
  \medskip
} \fi

\bigskip
\begin{abstract}
We propose a new algorithm for sparse estimation of eigenvectors in generalized eigenvalue problems (GEP). The GEP arises in a number of modern data-analytic situations and statistical methods, including principal component analysis (PCA), multiclass linear discriminant analysis (LDA), canonical correlation analysis (CCA), sufficient dimension reduction (SDR) and invariant co-ordinate selection. We propose to modify the standard generalized orthogonal iteration with a sparsity-inducing penalty for the eigenvectors. To achieve this goal, we generalize the equation-solving step of orthogonal iteration to a penalized convex optimization problem. The resulting algorithm, called penalized orthogonal iteration, provides accurate estimation of the true eigenspace, when it is sparse. Also proposed is a computationally more efficient alternative, which works well for PCA and LDA problems. Numerical studies reveal that the proposed algorithms are competitive, and that our tuning procedure works well. We demonstrate applications of the proposed algorithm to obtain sparse estimates for PCA, multiclass LDA, CCA and SDR. Supplementary materials are available online.
\end{abstract}

\noindent%
\vfill

\newpage
\spacingset{1.45} 

\section{Introduction} \label{sec:intro}

Many statistical problems can be cast into the mathematical framework of a generalized eigenvalue problem (GEP). In particular, we focus on the symmetric-definite GEP, posed as follows.
Suppose that $\Av \in \mathbb R^{p\times p}$ is a symmetric matrix, and $\Bv \in \mathbb R^{p\times p}$ is a symmetric positive-definite matrix.
While specific statistical contexts determine the exact nature of $\Av$ and $\Bv$,
a solution of GEP is given by a $d$-dimensional subspace $\mathcal U_d$ that is spanned by the generalized eigenvectors $\uv_1, \ldots, \uv_d$ corresponding to the $d$ largest generalized eigenvalues.
The following equations define a generalized eigen-pair $(\lambda_j, \uv_j)$:
\begin{equation}\label{eq:gep0}
\Av\uv_j = \lambda_j \Bv\uv_j,
\end{equation}
where the generalized eigenvalues, $\lambda_1 \ge \cdots \ge \lambda_d$, satisfy $\lambda_j = \uv_j^\smt\Av\uv_j /\uv_j^\smt\Bv\uv_j$.
The generalized eigenvectors are orthogonal with respect to $\Bv$, i.e., $\uv_i^\smt\Bv\uv_j = 1$ for $i = j$, and $0$ for $i\neq j$.  In the special case of $\Bv = \Iv_p$, the GEP is the standard eigenvalue problem.

Immediate applications of the GEP are to multivariate analysis problems, including the principal component analysis (PCA), canonical correlation analysis (CCA), multiclass linear discriminant analysis (LDA), invariant co-ordinate selection \citep{tyler2009invariant} and sufficient dimension reduction \citep{li2007sparse}.
The GEP also appears frequently in nonlinear dimension reduction \citep{kokiopoulou2011trace} and in computer vision and image processing \citep{zhang2013trace}.
We refer to the online supplementary material for detailed description of some selected statistical GEP problems.  

When $p$ is large, it is often desirable to find a sparse representation of $\Uc_d$, so that its basis vectors are sparse, i.e., the vectors have many zero loadings in their entries.
We propose an efficient algorithm for estimating sparse generalized eigenvectors from noisy observations of $\Av$ and $\Bv$.
It is well-known that in the standard eigenvalue problem, e.g. for PCA, the standard eigenvectors of the sample covariance matrix $\Av$ (while $\Bv = \Iv_p$) is inconsistent with the true principal directions for large $p$ \citep{Johnstone2009,Jung2009a}. Thus in the high-dimension, low-sample-size situations, sparsity-inducing methods have a potential in improving the estimation accuracy, as well as in providing interpretable eigenvectors through variable selection.

Our approach to obtain a sparse solution of the GEP is to extend the generalized orthogonal iteration, a standard numerical method of solving the GEP \citep{golub2012matrix}. The generalized orthogonal iteration consists of iterating two steps: one involving solutions of linear equations, and an orthogonalization step. One of our main ideas is to transform the equation-solving step to a minimization of a quadratic objective function, so that a sparsity-inducing penalty can be easily incorporated. We propose to use $\ell_1$-norm or $\ell_{2,1}$-norm penalty on the eigenvector matrix to induce element-wise or coordinate-wise sparsity. These penalty functions are those used in lasso and group-lasso regressions \citep[\emph{cf.}][]{Hastie2009}.
The proposed method is called \emph{Penalized Orthogonal Iteration} (POI). We utilize the block coordinate descent algorithm in solving the penalized minimization problem at each iteration.
We also study a computationally efficient alternative to POI, that comes down to solving just one minimization problem. We establish sufficient conditions under which the solutions of this alternative method, called \emph{Fast POI}, correspond to the solutions of the GEP. The conditions are satisfied under the situations for PCA and multiclass LDA.

The solutions of POI and Fast POI depend on the choice of a tuning parameter, dictating the degrees of penalization. Larger values of the tuning parameter result in more sparse solutions of the generalized eigenvectors. When the tuning parameter is zero, POI becomes the generalized orthogonal iteration. An eigenvalue-based cross-validation procedure is proposed, and is seen to work well in numerical studies.

To the best of authors' knowledge, there have been only a few proposals for sparse GEP in the literature \citep{sriperumbudur2011majorization,song2015sparse,tan2016sparse,gaynanova2016penalized,HanClemmensen:2016,Safo2018}. We briefly discuss their approaches and limitations, compared to our proposal.

Since finding the largest eigenvalue $\lambda_1$ that solves (\ref{eq:gep0}) is equivalent to maximize $\uv^\smt\Av\uv$ subject to $\uv^\smt\Bv\uv = 1$, \cite{sriperumbudur2011majorization} formulated both an $\ell_0$-constrained GEP and an $\ell_0$-penalized GEP:
\begin{align}
\max_\uv \uv^\smt\Av\uv, & \quad \mbox{subject to} \ \uv^\smt\Bv\uv = 1, \|\uv\|_0 \le s, \label{eq:constraint-gep-TAN}\\
\max_\uv \uv^\smt\Av\uv - \rho \|\uv\|_0, & \quad \mbox{subject to} \ \uv^\smt\Bv\uv = 1,  \label{eq:penalize-gep-Song}
\end{align}
for $1\le s \le p$ and $\rho > 0$, where $\|\uv\|_0$ is the number of nonzero elements of $\uv$. While \cite{sriperumbudur2011majorization} only proposed an algorithm to solve a variant of (\ref{eq:penalize-gep-Song}), \citet{song2015sparse} later proposed several approximate solutions of (\ref{eq:penalize-gep-Song}). 
Recently, \cite{tan2016sparse} proposed to solve (\ref{eq:constraint-gep-TAN}), by truncating the steepest ascent iterates in maximizing the Rayleigh coefficient $\uv \mapsto \uv^\smt\Av\uv /\uv^\smt\Bv\uv$.
\citet{gaynanova2016penalized} pointed out a fundamental difference between the penalized and constrained optimizations for sparse GEP, similar to (\ref{eq:constraint-gep-TAN}) and (\ref{eq:penalize-gep-Song}) but with $\ell_1$-norm.
\cite{Safo2018} proposed to estimate $\uv$ via minimizing $\norm{\uv}_1$ subject to a constraint $\|\Av \tilde\uv - \tilde\lambda \Bv \uv\|_\infty \le \rho$, where $(\tilde\lambda,\tilde\uv)$ is the non-sparse solution of (\ref{eq:gep0}).
As it is evident in \cite{sriperumbudur2011majorization}, \cite{song2015sparse}, and \cite{tan2016sparse} who limit themselves for solving only one eigen-pair, we are unclear how (\ref{eq:constraint-gep-TAN}) or (\ref{eq:penalize-gep-Song}) generalizes to simultaneously solving for multiple eigenvectors, $\uv_1,\ldots, \uv_d$.
 When multiple eigenvectors are needed, as is typical in practice, these methods are not readily applicable, at least not without a clever modification. Our algorithm is designed to estimate $\uv_1,\ldots, \uv_d$ altogether, and works well when $d > 1$.

\citet{HanClemmensen:2016} assumed $\Bv$ to be positive definite, and transformed the GEP to a regular eigen-decomposition of $\Bv^{-1}\Av$ (or $\Bv^{-1/2}\Av\Bv^{-1/2}$) while applying an $\ell_1$-penalty to achieve sparsity. However, their method is not directly applicable to the large-$p$-small-$n$-case, due to the numerically unstable inverse of the large matrix $\Bv$.
They used alternating direction method of multipliers for optimization, which causes their method to be computationally expensive.


\cite{chen2010coordinate} proposed to solve sufficient dimension reduction (SDR) problems by maximizing
$\mbox{trace}(\Uv^\smt\Av\Uv) - \rho_\lambda(\Uv)$,
for $\Uv \in \Real^{p\times d}$ satisfying $\Uv^\smt\Bv\Uv = \Iv_d$, in which the penalty function $\rho_\lambda$  enforces coordinate-wise sparsity (\ref{eq:estimation1_GROUPLASSO}). 
While \cite{chen2010coordinate}'s formulation is similar to our Fast POI with (\ref{eq:estimation1_GROUPLASSO}), their computation is much slower than any of our proposed algorithms, perhaps due to using both penalization and constraint. 

Our proposed algorithms provide sparse solutions of the original GEP (\ref{eq:gep0}), produces any number of eigenpairs simultaneously, is computationally efficient even for high-dimensional data, and is applicable to a number of statistical problems including PCA, LDA, CCA, SDR and invariant co-ordinate selection.

The rest of the paper is organized as follows.
In Section \ref{sec:method}, we introduce the   proposed sparse GEP methodology. Numerical algorithms are discussed in Section \ref{sec:algorithm}.
We demonstrate applications of our proposal to a number of statistical problems including PCA, LDA, SDR and CCA in Section \ref{sec:simulations}, in which some of the most promising competitors are numerically compared.
All proofs are contained in the appendix. The online supplementary material contains additional numerical results.

\section{Methodology}\label{sec:method}

\subsection{Setting}\label{sec:setting}

In most applications, the population matrices $\Av$ and $\Bv$  in (\ref{eq:gep0}) are symmetric non-negative definite, while $\Bv$ is often positive definite. For some applications, the rank of $\Av$ is much smaller than the dimension $p$ of the matrices.

We are interested in estimating the generalized eigenvectors $\uv_1, \ldots, \uv_d$ corresponding to the $d$ largest generalized eigenvalues, where $d \le \mbox{rank}(\Av)$. Our proposed methods estimate the subspace $\Uc_d = \text{span}(\uv_1, \ldots, \uv_d)$.
In order for $\Uc_d$ to be identifiable, one must assume the following on the size of the eigenvalues:
\begin{equation*}\label{eq:ev_condition}
\lambda_1 \ge \cdots \ge \lambda_d > \lambda_{d+1} \ge \cdots \ge \lambda_{p} \ge 0.
\end{equation*}

Let $\Uv_d = [\uv_1,\ldots,\uv_d]$ and $\Lambdav_d = \mbox{diag}(\lambda_1,\ldots,\lambda_d)$. The generalized eigenvectors are $\Bv$-orthogonal, and thus in general do not form an orthogonal basis of $\Uc_d$. However, one can readily obtain an orthogonal basis of $\Uc_d$ from $\Uv_d$, and, conversely, obtain $\Uv_d$  from any orthogonal basis of $\Uc_d$. The detail follows. Throughout the paper, the notation $\Oc(p,d)$ is used for the set of semi-orthogonal matrices; $\Oc(p,d) = \{\Xv \in \Real^{p\times d} :  \Xv^\smt \Xv = \Iv_d\}$. Let $\Oc(d) = \Oc(d,d)$ be the set of $d\times d$ orthogonal matrices.

The following equation is equivalent to (\ref{eq:gep0}):
\begin{equation}\label{eq:gepU}
\Av \Uv_d = \Bv \Uv_d \Lambdav_d.
\end{equation}
By the QR decomposition, we have $\Uv_d = \Qv_d\Rv_d$, for a $\Qv_d \in \Oc(p,d)$, 
and for a $d \times d$ upper-triangular matrix $\Rv_d$.
Then the GEP in (\ref{eq:gepU}) is equivalently written as
\begin{equation}\label{eq:gepQ}
\Av \Qv_d = \Bv \Qv_d \Lambdav_d^*,
\end{equation}
where $\Lambdav^*_d = \Rv_d\Lambdav_d\Rv_d^{-1}$ is a $d\times d $ upper-triangular matrix. Note that the diagonal elements of $\Lambdav_d^*$ are the same as those of $\Lambdav_d$.

Now suppose we have an arbitrary basis matrix $\tilde\Qv_d \in \Oc(p,d)$ such that $\mbox{span}(\tilde\Qv_d) = \Uc_d$. The generalized eigenvectors satisfying (\ref{eq:gepU}) can be recovered, as follows:

\begin{prop}\label{thm1}
Assume that $\tilde\Qv_d$ is an arbitrary $p \times d$ orthogonal matrix with the same column space as $\Uv_d$, where ($\Lambdav_d,\Uv_d)$ is the solution to (\ref{eq:gepU}).
For $\tilde{\Av}_{d} =  \tilde\Qv_d^\smt\Av\tilde\Qv_d$ and $\tilde{\Bv}_d = \tilde\Qv_d^\smt\Bv\tilde\Qv_d$,
 let $\Tv$ and $\Dv$ respectively be the matrix of eigenvectors and the diagonal matrix of eigenvalues of the following GEP:
\begin{equation}\label{eq:gepdual}
\tilde{\Av}_{d}\Tv = \tilde{\Bv}_{d}\Tv \Dv
\end{equation}
with $\Tv^\smt\tilde{\Bv}_{d}\Tv= \Iv_d$.
Then $\mbox{span}(\Uv_d) = \mbox{span}(\tilde\Qv_d\Tv)$.
If the diagonal values of $\Dv$ are distinct and in the decreasing order, then $\Uv_d = \tilde\Qv_d\Tv$ and $\Lambdav_d = \Dv $.
\end{prop}

In the next subsection, we discuss our approaches in estimating an orthogonal basis of $\Uc_d$, i.e., $\Qv_d$ in (\ref{eq:gepQ}), from noisy versions of $\Av$ and $\Bv$. Proposition~\ref{thm1} can then be used to obtain the estimates of the generalized eigenvector and eigenvalue pair ($\Uv_d$, $\Lambdav_d$).
In practice the matrices $\Av$ and $\Bv$ in (\ref{eq:gepU}) are replaced by their empirical counterparts computed from a sample.
We treat $\Av$ and $\Bv$ to be the empirical matrices, with which we attempt to solve the GEP.
In most applications, the empirical matrices $\Av$ and $\Bv$ are non-negative definite by construction. We require $\Bv$ to be positive definite.
If not, we add a scaled identity matrix $\epsilon\Iv_p$, for a small $\epsilon>0$, and $\Bv + \epsilon\Iv_p$ is treated as $\Bv$.
We recommend using $\epsilon = \min(\log p / \mbox{rank}(\Bv), \sigma_{\Bv}/2$) where the $\sigma_{\Bv}$ is the smallest positive eigenvalue of $\Bv$.

\subsection{Proposed Method}\label{sec:proposed}
We consider two penalized optimization approaches in order to obtain sparse solutions of the GEP. The first approach is a generalization of the widely used orthogonal iteration and can be applied to almost all situations. The second approach aims to provide efficient computation for some high-dimensional problems, such as PCA and multiclass LDA.

\subsubsection{Penalized Orthogonal Iteration}\label{sec:POI}

We begin by reviewing the standard generalized orthogonal iteration in solving the GEP (\ref{eq:gepQ}).
Given $(\Av,\Bv)$ and an initial value $\hat\Qv_0 \in \Oc(p,d)$,
the standard generalized orthogonal iteration \citep[Section 8.7.3,][]{golub2012matrix} finds $\hat\Qv$, a solution of (\ref{eq:gepQ}), by iterating the following two steps until convergence. For $r = 1,2,\ldots$,

\textbf{Step 1.} Solve $\Bv \hat\Zv_r = \Av \hat\Qv_{r-1}$ for $\hat\Zv_r$.

\textbf{Step 2.} Obtain $\hat\Qv_r$ by QR decomposition $\hat\Zv_r = \hat\Qv_r \hat\Rv_r$.

The iteration stops when a distance between $\hat\Qv_{r-1}$ and $\hat\Qv_{r}$ is smaller than a threshold. We use the projection distance between the subspaces, defined later in (\ref{eq:projectionmetric}). 

The penalized orthogonal iteration (POI) modifies {Step 1} so that it can be cast into a convex optimization framework with a sparsity-inducing penalty. Let $p_\lambda(\Zv)$ be a penalty function on the matrix $\Zv$. Here, $\lambda \ge 0 $ represents the degrees of penalization, yielding $p_\lambda(\Zv) = 0$ if $\lambda = 0$. We propose to replace Step 1 by
\begin{equation}\label{eq:estimation1}
\hat\Zv_r = \argmin_{\Zv \in \Real^{p\times d}} \left\{
  \tr \left( \frac{1}{2} \Zv^\smt \Bv \Zv - \Zv^\smt\Av \hat\Qv_{r-1} \right) + p_\lambda (\Zv)
 \right\}.
\end{equation}
Note that if $\lambda = 0$, then  $\hat\Zv_r $ of (\ref{eq:estimation1}) is the solution of the original equation:
\begin{equation}\label{eq:estimation1-linsystem}
\Bv \hat\Zv_r = \Av \hat\Qv_{r-1}.
\end{equation}

The POI optimization problem (\ref{eq:estimation1}) is motivated by the fact that the linear equation system (\ref{eq:estimation1-linsystem}) is the first-order condition of a quadratic optimization problem (\ref{eq:estimation1}) without the penalty term. 

A natural choice for the sparsity-inducing penalty function $p_\lambda(\Zv)$ in (\ref{eq:estimation1}) is a lasso penalty \citep{tibshirani1996regression},
\begin{equation}\label{eq:estimation1_LASSO}
p_\lambda(\Zv) = \sum_{j=1}^d \lambda_j \norm{\zv_j}_1,
\end{equation}
where $\Zv = [\zv_1,\ldots,\zv_d]$ and $\lambda_j>0$.
While  (\ref{eq:estimation1_LASSO}) generally produces \emph{element-wise sparse} solutions, it may be more reasonable to assume that there exists a small subset of coordinate indices that are relevant to all of the $d$ largest generalized eigenvalues. For this purpose, we use the $\ell_{2,1}$-norm penalty for $p_\lambda (\Zv)$, also known as the group-lasso penalty \citep{yuan2006model}.  Let $\z_i^\smt$ be the $i$th row of $\Zv$ ($i =1,\ldots, p$). Then we set for $\lambda > 0$,
\begin{equation}\label{eq:estimation1_GROUPLASSO}
p_\lambda(\Zv) = \lambda \sum_{i=1}^p  \norm{\z_i}_2.
\end{equation}
We have implemented the POI with the penalty (\ref{eq:estimation1_LASSO}) for element-wise sparsity and (\ref{eq:estimation1_GROUPLASSO}) for \emph{coordinate-wise sparsity}. Other choices of penalty functions can be used as well; see, for example, \cite{tibshirani2011regression}.

In search of a sparse subspace basis, coordinate-wise sparsity is preferred to element-wise sparsity, since any element-wise sparse basis matrix $\Qv$ is in general no longer  element-wise sparse if arbitrarily rotated, e.g., in the orthogonalization step of the POI. In contrast, $\Qv\Vv$ for any $\Vv \in \Oc(d)$  is coordinate-wise sparse as long as $\Qv$ is coordinate-wise sparse. Since the basis of a subspace can only be coordinate-wise sparse, utilizing the coordinate-wise sparsity has a clear advantage in variable selection and its interpretation \citep[\emph{cf}. ][]{Bouveyron2016}.


\begin{remark}
In the PCA context (i.e, $\Bv = \Iv_p$ and $\Av = \widehat{\boldsymbol{\Sigma}}$, the sample covariance matrix), the POI overlaps with the sparse PCA algorithm of \cite{ma2013sparse}, who proposed to add a thresholding step to the orthogonal iteration for the standard eigenvalue problem. In particular, if the penalty function $p_\lambda$ is given by an $\ell_1$-norm, then our method coincides with Ma's method with soft-thresholding.
\end{remark}
%

\subsubsection{Fast POI}\label{sec:FPOI}

A computationally simpler alternative to the POI is to solve an alternative form of  (\ref{eq:estimation1}) just once. 
For this, we propose to replace $\Av\Qv_{r-1}$ in (\ref{eq:estimation1}) with the $p\times d$ matrix, denoted by $\Vv$, whose columns contain the eigenvectors of $\Av$ that correspond to the $d$ largest eigenvalues of $\Av$. This results in the following:
  \begin{equation}\label{eq:fpoi}
\hat\Zv = \argmin_{\Zv \in \Real^{p\times d}} \left\{
  \tr \left( \frac{1}{2} \Zv^\smt \Bv \Zv - \Zv^\smt\Vv \right) + p_\lambda(\Zv)
 \right\}.
\end{equation}
The orthogonal basis matrix $\hat\Qv$ is obtained by the QR decomposition of $\hat\Zv$. There is no ``outer'' iteration for this  approach, which yields much faster computation. The penalty function $p_\lambda(\Zv)$ can be chosen to be either (\ref{eq:estimation1_LASSO}) or (\ref{eq:estimation1_GROUPLASSO}). The solutions from this approach is referred to as \emph{Fast POI} solutions.

Note that the  POI solves the GEP in (\ref{eq:gepQ}) if $\lambda = 0$. On the other hand, the Fast POI solution with $\lambda = 0$ does not in general solve the original GEP. The following proposition identifies a sufficient condition under which the Fast POI solution with no penalty solves the original GEP.


\begin{prop} \label{prop:fpoi}\label{lem:optionDsubspaceequality}
Suppose $\Av$ and $\Bv$ are both $p\times p$ symmetric non-negative matrices. For any $1 \le d < p $, let $\Vv \in \Oc(p,d)$ be the eigenvector-matrix of $\Av$ corresponding to the $d$ largest eigenvalues of $\Av$, and $\Qv \in \Oc(p,d)$ satisfy the generalized eigen-equation (\ref{eq:gepQ}).
   If
  \begin{enumerate}
    \item[(a)] $\Bv = \Iv_p$, or
    \item[(b)] $\Bv$ is positive definite and $\rank(\Av) = d$,
  \end{enumerate}
  then the column space of $\Bv^{-1}\Vv$ is exactly the column space of $\Qv$.
\end{prop}

The condition (a) corresponds to the standard eigenvalue problem, arising in the context of PCA. The rank condition in (b) is satisfied when there are $K = d+1$ groups in a multiclass LDA problem.

In our experience, obtaining a Fast POI solution requires only a fraction of time that is needed for the corresponding POI solution. This is because POI typically requires many iterations until convergence. Despite its quick computing time, we found that Fast POI in the problems of PCA and multiclass LDA provides good approximations of the POI solutions.

\begin{remark}\label{remark:1-poi-Mai-Gaynanova}
After we formulated the Fast POI criterion, we found that, in the special case of multiclass LDA problems, our minimization approach (\ref{eq:fpoi}) is similar to the sparse LDA methods of \cite{mai2015multiclass} and \cite{gaynanova2016simultaneous}.
 Using the notation used in (\ref{eq:fpoi}), Mai et al's solutions are given by replacing $\Vv$ with the mean difference matrix with columns $\hat{\mu}_k -\hat{\mu}_1$, $k=2,\ldots,K$; Gaynanova et al's approach is equivalent to (\ref{eq:fpoi}) if $-\Zv^\smt\Vv$ is replaced by $\frac{1}{2}\Zv^\smt\Av\Zv -   \Zv^\smt\Dv$, where $\Dv$ consists of weighted mean differences $\hat{\mu}_i -\hat{\mu}_j$. Both approaches are dependent on the order of group labels and, consequently, their solutions change if different label orders are used. 
 In Gaynanova et al's approach, the inclusion of $\frac{1}{2}\Zv^\smt\Av\Zv$ is to ensure the objective function to be bounded below. Since we use $\Bv + \epsilon\Iv_p$ (if $\Bv$ is not full rank), the objective function (\ref{eq:fpoi}) is bounded. These methods are further compared numerically in Section~\ref{sec:numerical_LDA}.
 \end{remark}

\subsection{Eigenvalue Estimation}

Once $\hat\Qv$  is obtained either by POI or Fast POI, one might be interested in estimating the corresponding generalized eigenvalues. We illustrate two approaches here.
First, one can estimate the generalized eigenvalues by directly using Proposition \ref{thm1}. Let $\hat\Tv$ and $\hat\Dv$ be the solution to the GEP (\ref{eq:gepdual}) with $\Qv$ replaced by $\hat\Qv$.  Then the plug-in estimate of the eigenvector matrix is $\hat\Uv=\hat\Qv\hat\Tv$, and the eigenvalue matrix estimate is $\hat \Lambdav= \hat \Dv$.
If the elements of $\Av$ stand for covariances,
then using Proposition \ref{thm1} is  equivalent to
estimating $\lambda_j$ by the empirical variance contained in $\Av$ projected on $\hat\uv_j$, normalized by that of $\Bv$:
\begin{equation}
\label{eq:sample_eigenequation0}
 \hat \lambda_j =  { \hat\uv_j^\smt \Av   \hat \uv_j }/{ \hat \uv_j^\smt \Bv   \hat \uv_j} ,~ j=1,\dots,d.
\end{equation}

A more sophisticated approach is to use the eigen-equation (\ref{eq:gepU})  directly. For this, we estimate the generalized eigenvalue matrix $\Lambdav$ 
by the solution of the following minimization problem:
\begin{align}
 \tilde\Lambdav =  \argmin_{\Lambdav}
  & \norm{\Av \hat\Uv - \Bv \hat\Uv \Lambdav}_F^2
  , \quad (\Lambdav \mbox{ is diagonal}).  \label{eq:sample_eigenequation1sol}
\end{align}
This problem has a closed-form solution. 

\begin{lem}\label{lem:closeformsolutionforEvalues}
 Let $\alphav_j$ and $\betav_j$ be the $j$th column vectors of $\Av \hat\Uv$ and $\Bv \hat\Uv$. Then the solution $\tilde\Lambdav$ of (\ref{eq:sample_eigenequation1sol}) is $\tilde\Lambdav = \mbox{diag}(\tilde\lambda_1,\ldots,\tilde\lambda_d)$, where
 $\tilde\lambda_j = (\betav_j^\smt \betav_j)^{-1} \betav_j^\smt \alphav_j$.
%
\end{lem}

The eigenvalue estimates $\hat\Lambdav$ of  (\ref{eq:sample_eigenequation0}) and $\tilde\Lambdav$ of (\ref{eq:sample_eigenequation1sol}) are in general different. They are the same when $\hat{\uv}_j$ are the exact solutions of the GEP (\ref{eq:gepU}), which are generally non-sparse.
In our experiments, these two estimates were numerically close to each other. We used (\ref{eq:sample_eigenequation0}) in all of our numerical analyses.

\subsection{Tuning Parameter Selection}\label{sec:tuning}

The tuning parameter $\lambda$ in the penalty function plays an important role in the estimation. When $\lambda = 0$, our proposal provides non-sparse solutions, while the estimated subspaces tend to have more sparse bases as $\lambda$ increases.
We consider a general cross-validation approach that can be used for any regularized GEP.

Suppose that we have a data set that can be split into two subsets that produce independent pairs $(\Av^{(r)}, \Bv^{(r)})$, $r = 1,2$. We use the data indexed by $r = 1$ to train the estimates $\hat\Uv_{\lambda}$ and $\hat\Lambdav_{\lambda}$ for various values of $\lambda$.
The data indexed by $r = 2$ is used to evaluate these estimates, to tune $\lambda$. In particular, for a given $\lambda$, we define the cross-validation score as
\begin{equation}\label{eq:CV2}
CV(\lambda) = \mbox{trace}\left[ (\hat\Uv_\lambda^\smt \Bv^{(2)} \hat\Uv_\lambda )^{-1}  \hat\Uv_\lambda^\smt \Av^{(2)} \hat\Uv_\lambda\right].
\end{equation}

We choose the $\lambda$ with the largest average cross-validation score, based on repeated random splits, or by predefining the training and tuning sets.
Candidate values of $\lambda$ are found in an interval $(0,\lambda_{\max})$.

The upper bound $\lambda_{\max}$ is set so that $\hat\Uv_\lambda = \0v$ for any $\lambda > \lambda_{\max}$, and  depends on the data and the penalty term used.
%
%
%
We defer the discussion on the choice of $\lambda_{\max}$ to Section~\ref{sec:algorithm}; see Remarks~\ref{remark2} and \ref{remark3}.


We interpret the CV score (\ref{eq:CV2}) as follows. Note that when $\Bv^{(2)}$ and $\Av^{(2)}$ are replaced by $\Bv^{(1)}$ and $\Av^{(1)}$, $CV(\lambda)$ is simply the sum of estimated eigenvalues (\ref{eq:sample_eigenequation0}). Heuristically, maximizing $CV(\lambda)$ is equivalent to finding the subspace with the largest sum of ``prediction eigenvalues''.
%
%
We delve further into the following data-analytic situations.

For PCA,  $CV(\lambda)$ represents the total variance contained in the $d$-dimensional eigenspace. To be specific, let $\Xv$ be the centered $n \times p$ data matrix, then we set $\Av^{(2)} = \Xv^\smt \Xv$ and $\Bv^{(2)} = \Iv_p$, which in turn leads to
$$ CV(\lambda) = \mbox{trace}(\hat\Uv_\lambda^\smt \Av^{(2)} \hat\Uv_\lambda)
               = \mbox{trace}(\Xv^\smt \Xv) -\|\Xv (\Iv_p - \hat\Uv_\lambda \hat\Uv_\lambda^\smt)  \|^2_F ,$$
where the second term is sometimes called \emph{the reconstruction error} of the subspace spanned by $\hat\Uv_\lambda$.
 Thus maximizing $CV(\lambda)$ is equivalent to minimizing the reconstruction error, an approach that can be found in the literature on PCA; see e.g., \cite{shen2008sparse,josse2012selecting}.
In our numerical studies, we found that $CV(\lambda)$ is typically a concave function, and is negatively correlated with the distance to the true eigenspace from the estimate associated with $\lambda$, i.e. the larger $CV(\lambda)$, the better the estimation. These are illustrated in Fig.~~\ref{fig:tuning} (top two panels) with simulated data, further discussed in Section~\ref{sec:numerical}.

\begin{figure}[p]
  \centering
  \includegraphics[width=0.8\textwidth]{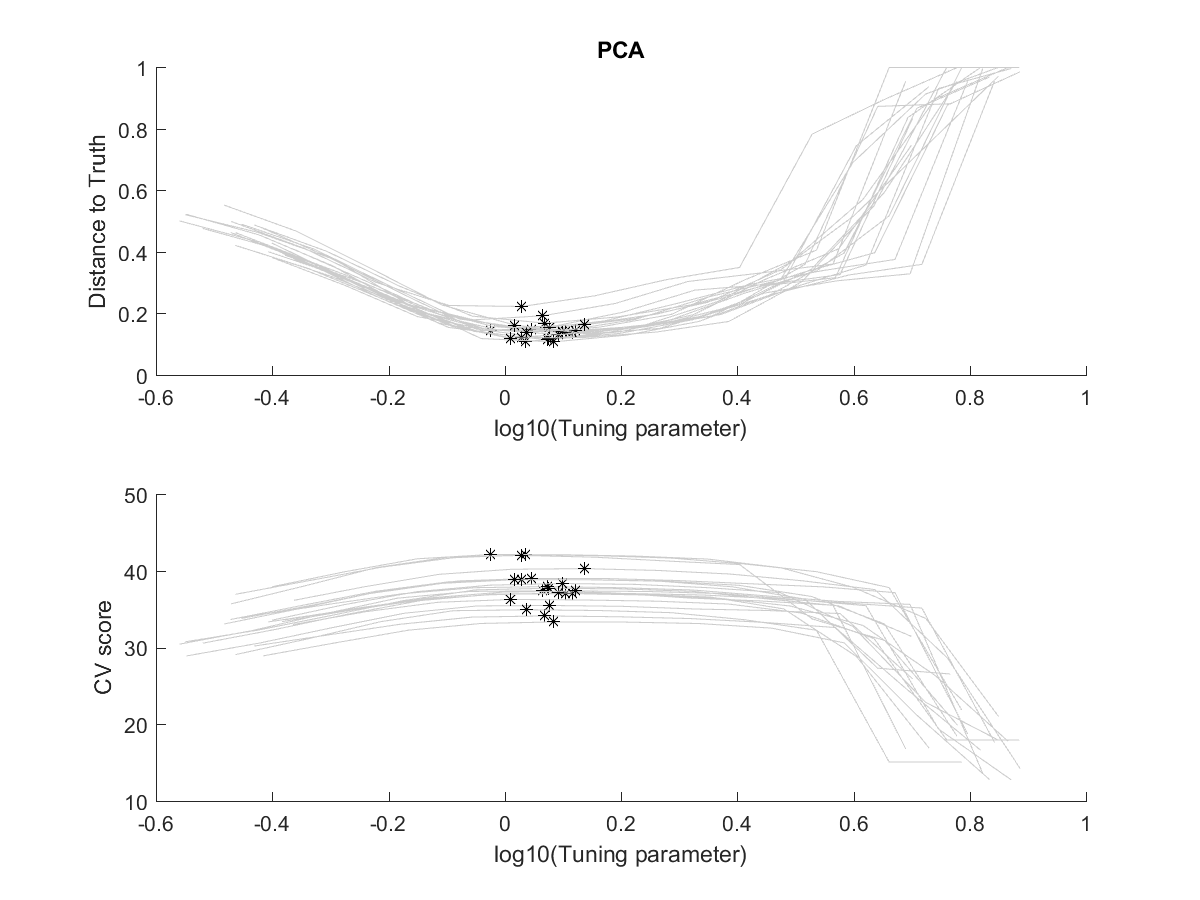}\\
  \includegraphics[width=0.8\textwidth]{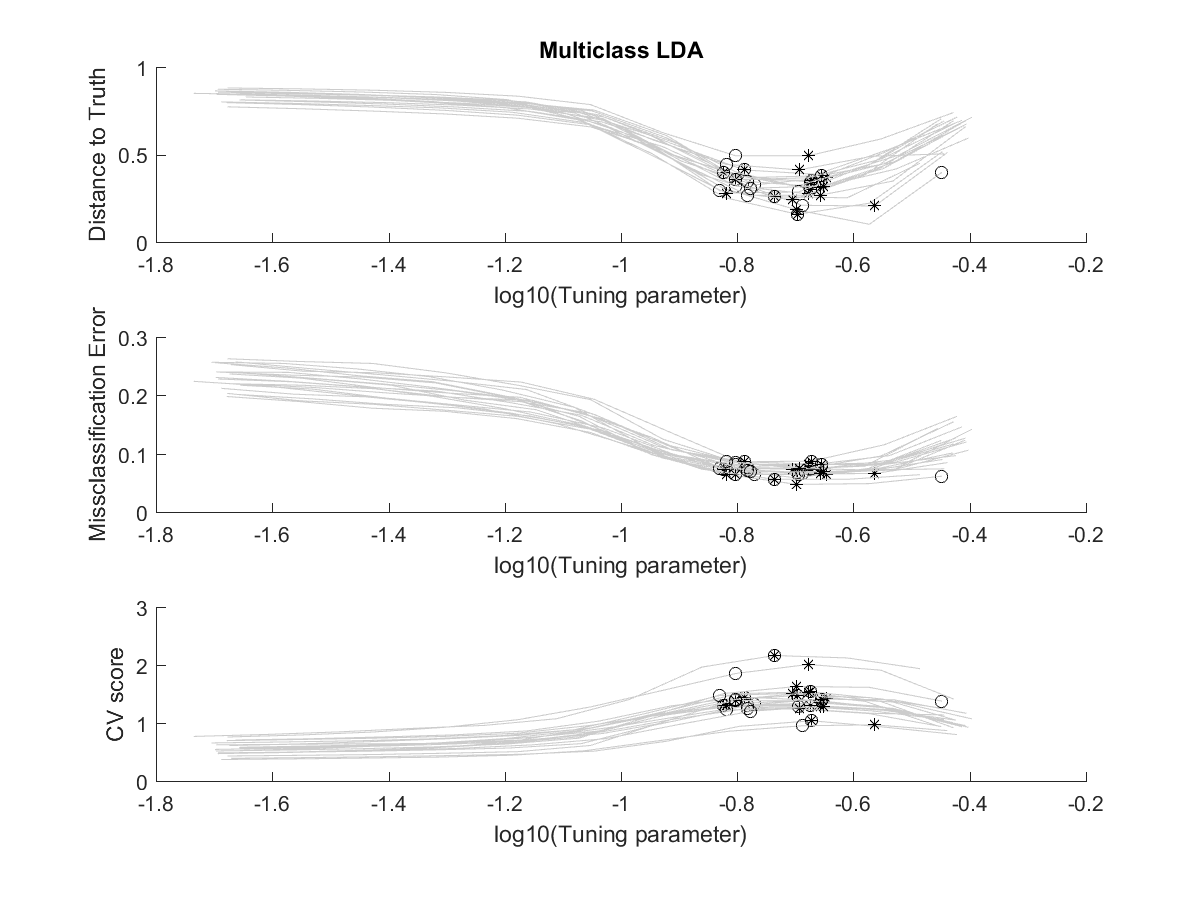}\\
  \caption{The proposed tuning procedure works well for our simulated data. Top two panels are from PCA models; bottom three panels are from multiclass LDA models.  $*$ indicates the location of tuned $\lambda$ that maximizes $CV(\lambda)$. $\circ$ indicates the location of tuned $\lambda$ that minimizes the tuning misclassification error rate.}\label{fig:tuning}
\end{figure}

In a situation with multiple groups (\emph{cf}. Section~\ref{sec:numerical_LDA}), we set $\Av^{(2)} = \Sv_B$ as the between-group covariance, and $\Bv^{(2)} = \Sv_W$ as the within-group covariance, estimated from the tuning set. The groups are more clearly separated for larger $\Sv_B$ and smaller $\Sv_W$.
Thus, $CV(\lambda)$ is large if the groups, projected on the subspace spanned by $\hat\Uv_\lambda$,  are well-separated. In a view from clustering, $CV(\lambda)$ is akin to the so-called CH index \citep{calinski1974dendrite} in spirit; they are proportional to each other if $d = 1$. Larger CH index indicates clearer clustering, which is also associated with larger $CV(\lambda)$. Figure~\ref{fig:tuning} (bottom three panels) shows that, in the case of multiclass LDA, $CV(\lambda)$ is a concave function of $\lambda$ and is negatively correlated with the misclassification rate. This is not unexpected, because larger $CV(\lambda)$ implies clearer separation of groups, which in turn makes the classification easier.

 For classification in mind, one may use the misclassification rate as a tuning creterion. 
Tuning by (\ref{eq:CV2}) is in fact on par with tuning by misclassification rate; see Section~\ref{sec:numerical_LDA}. The tuned $\lambda$ by both methods are close to each other, as shown in Fig.~\ref{fig:tuning}. 

We note that since our tuning procedure is intended for eigenspace prediction accuracy, it sometimes chooses more variables than desired. For more precise variable selection, one may adopt the one-standard error rule; choose the most parsimonious model within one standard error of the maximum $CV(\lambda)$, which may only be estimated under multiple splits of data as in $K$-fold cross-validation.

\section{Algorithms}\label{sec:algorithm}
The POI and Fast POI solutions can be efficiently implemented. We provide algorithms to solve the POI (\ref{eq:estimation1}) when the penalty function is either (\ref{eq:estimation1_LASSO}) or (\ref{eq:estimation1_GROUPLASSO}). Note that the Fast POI (\ref{eq:fpoi}) can be solved in the same way (\ref{eq:estimation1}) is solved, with $\Av \Qv_{r-1}$ replaced by ${\Vv}$.  Our algorithms for (\ref{eq:estimation1}) are guaranteed to converge to the optimal solution \citep[\emph{cf}.][]{tseng1993dual}, thus the Fast POI always converges.
We are not aware of a general condition under which the proposed POI is  guaranteed to converge, although it has converged in all of our experiments. 

\subsection{Algorithm for element-wise sparse estimation}

Solving (\ref{eq:estimation1}) with the penalty (\ref{eq:estimation1_LASSO}) amounts to solving $d$ separate problems. Specifically, the $j$th column $\zv_j$ of $\Zv_r$ is
\begin{equation}\label{eq:estimation2}
\zv_j = \argmin_{\zv \in \Real^p}
\left\{
  \left( \frac{1}{2} \zv^\smt \Bv \zv - \zv^\smt\Av \qv_j \right) + \lambda_j \norm{\zv}_1
 \right\},
\end{equation}
where $\qv_j$ is the $j$th column of $\Qv_{r-1}$.
The minimization problem (\ref{eq:estimation2}) can be efficiently solved by the cyclic coordinate descent algorithm. The algorithm updates the coordinates of the iterate $\zv$ in a cyclic fashion. In particular, let $S_\lambda(z) = \mbox{sign}(z) (|z| - \lambda)_+$ be the soft-thresholding operator. Then, the $i$th coordinate $z_i$ of $\zv$ is updated by
\begin{equation}\label{eq:estimation_softthresholding}
z_i \leftarrow  S_{\lambda_j}\left(  \av_i^\smt \qv_j - \bv_i^\smt \zv + {b}_{ii} z_i
  \right)/{b}_{ii},
\end{equation}
where ${b}_{ii}$ is the $i$th diagonal element of $\Bv$, and $\av^\smt_i$ (or $\bv^\smt_i$) is the $i$th row of the symmetric matrix $\Av$ (or $\Bv$, respectively).


\begin{remark}\label{remark2}

 For both POI and Fast POI,  setting too large $\lambda_j$ gives the trivial solution $\zv_j = 0$.
 It can be shown that $\lambda^{o}_{j} = \max_{i = 1,\ldots,p} | \av_i^\smt \qv_{j} |$ is the maximum value of $\lambda_{j}$ which gives a non-trivial solution of  (\ref{eq:estimation2}).
 However, $\lambda^{o}_{j}$ depends on the iterate $\hat\Qv_{r-1}$. Since  the solution $\zv_j$ of (\ref{eq:estimation2}) is sparse  for large $\lambda_j$, we take the maximum of $\lambda^{o}_{j}$ over the possible values of ``sparse'' $\Qv_{r-1} \in \Oc(p,d)$ such that each column of $\Qv_{r-1}$ has only one nonzero entry. For $\Oc_0(p,d) = \{\Qv \in \Oc(p,d): \#\{(i,j): q_{ij} \neq 0 \} = d \}$, we set
 \begin{align}
 \lambda_{\max}
 = \max_{i = 1,\ldots,p} \left\{ \max_{\Qv  \in \Oc_0(p,d)} \max_{j=1,\ldots,d} |\av_i^\smt \qv_{j}|  \right\} = \max_{i,j} {|{a}_{ij}|},  \label{eq:A-lasso_MAX}
 \end{align}
 where $a_{ij} $ is the $(i,j)$th element of $\Av$. 
 \end{remark}

\subsection{Algorithm for coordinate-wise sparse estimation}
Since the problem (\ref{eq:estimation1}) is convex and the penalty function (\ref{eq:estimation1_GROUPLASSO}) is block-separable, the block coordinate descent algorithm is guaranteed to converge to an optimal solution \citep{tseng1993dual}.
Denote $\qv_j$ and $\q_i^\smt$ for the $j$th column and $i$th row of $\Qv_{r-1}$, respectively. Define $\zv_o^\smt = [\zv_1^\smt,\ldots,\zv_d^\smt]$ and $\qv_o^\smt = [\qv_1^\smt,\ldots,\qv_d^\smt]$, $\zv^\smt = [\z_1^\smt,\ldots,\z_p^\smt]$ and $\qv^\smt = [\q_1^\smt,\ldots, \q_p^\smt]$. Then,
\begin{align}
\tr \left( \frac{1}{2} \Zv^\smt \Bv \Zv - \Zv^\smt\Av \Qv_{r-1}  \right)
      &=  \sum_{j=1}^d \left(\frac{1}{2}\zv_j \Bv  \zv_j^\smt - \zv_j \Av \qv_j\right)  \nonumber \\
      &=  \frac{1}{2}\zv_o^\smt ( \Iv_d \otimes \Bv) \zv_o - \zv_o^\smt ( \Iv_d \otimes \Av) \qv_o  \nonumber \\
      &= \frac{1}{2}\zv^\smt (\Bv \otimes \Iv_d) \zv - \zv^\smt (\Av \otimes \Iv_d) \qv  \nonumber \\
      & = \frac{1}{2}\sum_{i=1}^p \sum_{j=1}^p \z_i^\smt b_{ij}\z_j -  \sum_{i=1}^p \z_i^\smt \deltav_i,  \label{eq:grlasso1}
\end{align}
where $a_{ij}$, (or $b_{ij}$) is the $(i,j)$th element of $\Av$ (or $\Bv$, respectively) and $\deltav_i = \sum_{j=1}^p a_{ij} \q_{j}$.
Using the fact that $\Bv$ is symmetric, for the $g$th block with all $\z_i$, $i \neq g$, fixed, the problem becomes
   \begin{equation}  \label{eq:c-lasso_problem}
   \min_{\z_g}
   \left\{
   \frac{1}{2} \z_g^\smt b_{gg}\z_g - \z_g^\smt \av_g  + \lambda \norm{\z_g}_2
   \right\},
   \end{equation}
where   $\av_g = \deltav_g - \sum_{i\neq g} b_{gi} \z_i$, and its solution is obtained by
   \begin{equation}  \label{eq:c-lasso_solutions}
    \hat\z_g = \frac{1}{b_{gg}}\left( 1 - \frac{\lambda}{\norm{\av_g}_2}\right)_+ \av_g.
   \end{equation}

In short, the algorithm updates $\Zv$ in a cyclic fashion, with the initial value given by $\Zv = \Qv_{r-1}$. For $g = 1,\ldots,p$, the $g$th row $\z_g^\smt$ of $\Zv$ is updated by (\ref{eq:c-lasso_solutions}) until a convergence criterion is met.

\begin{remark}\label{remark3}
 It is easy to see that $\lambda_{0} = \max_{g = 1,\ldots,p} \norm{\deltav_g}_2$ is the maximum value of $\lambda$ that gives a non-trivial solution of (\ref{eq:c-lasso_problem}). Since $\deltav_g = \sum_{j=1}^p a_{gj} \q_{j}$ depends on the iterate $\Qv_{r-1}$ of the POI, we take
 \begin{equation}  \label{eq:c-lasso_MAX}
 \lambda_{\max}  = \max_{g = 1,\ldots,p} \left\{ \max_{\Qv \in \Oc_0(p,d)} \norm{\sum_{j=1}^p a_{gj} \q_{j}}_2 \right\},
 \end{equation}
 which provides an \emph{upper bound} for the maximum value of $\lambda$. Note that in (\ref{eq:c-lasso_MAX}) we used $\Qv \in \Oc_0(p,d)$ whose entrywise $L_0$-norm is $d$ (that is, as sparse as possible).
To simplify (\ref{eq:c-lasso_MAX}), for each fixed $g$, denote $a_{g,(j)}$ for the $j$th largest element among the $g$th row of $\Av$, in the absolute value.
Then $\max_{\Qv \in \Oc_0(p,d)} \norm{\sum_{j=1}^p a_{gj} \q_{j}}_2 =  (\sum_{j=1}^d a_{g,(j)}^2 )^{1/2}, $
which in turn gives, for POI solutions,
$$ \lambda_{\max}  = \max_{g = 1,\ldots,p}\bigg(\sum_{j=1}^d a_{g,(j)}^2 \bigg)^{1/2}.$$
For an upper bound of the tuning parameter used in Fast POI, it is straightforward to see that
$ \lambda_{\max}  = \max_{g = 1,\ldots,p}(\sum_{j=1}^d v_{gj}^2 )^{1/2},$
where $v_{gj}$ is the $(g,j)$th element of $\Vv$.
%

\end{remark}

\section{Applications to Multivariate Analysis}
\label{sec:example}\label{sec:simulations}\label{sec:numerical}
In this section, we demonstrate the application of the proposed algorithms to principal component analysis (PCA), multiclass linear discriminant analysis (LDA), sufficient dimension reduction (SDR) and canonical correlation analysis (CCA). 


\subsection{Sparse Principal Component Analysis}\label{sec:numerical_PCA}

The standard PCA amounts to solving the ordinary eigen-decomposition of the covariance matrix $\boldsymbol{\Sigma}$.  By setting $\Av = \boldsymbol{\Sigma}$, $\Bv = \Id_p,$ the problem (\ref{eq:gep0}) becomes the ordinary eigen-decomposition problem, and the solution ($\uv_j$, $\lambda_j$)  corresponds to the principal component (PC) direction and variance pair.

We demonstrate the application of our method to sparse PCA, using a large-scale genomic data. The data set we use is adopted from \cite{ciriello2015comprehensive}, and consists of 272 cases of 1000 variables, where the first 500 variables are gene expression levels from the original data set, and the latter 500 are normally distributed random  noises. Since it is typical that in performing PCA one does not know which dimension $d$ to choose, we explore $d = 1$ to $20$, and estimate the principal subspace of dimension $d$ using POI with coordinate-wise sparse penalty.
A graphical summary of this study can be found in Fig.~\ref{fig:PCA_LF}. It can be seen that the eigenvalue and eigenvector estimates are stable across different choices of $d$, and our method has correctly screened out the noise variables for all choices of $d$. See the online supplementary material for details.

\begin{figure}[t]
  \centering
  \includegraphics[width=0.9\textwidth]{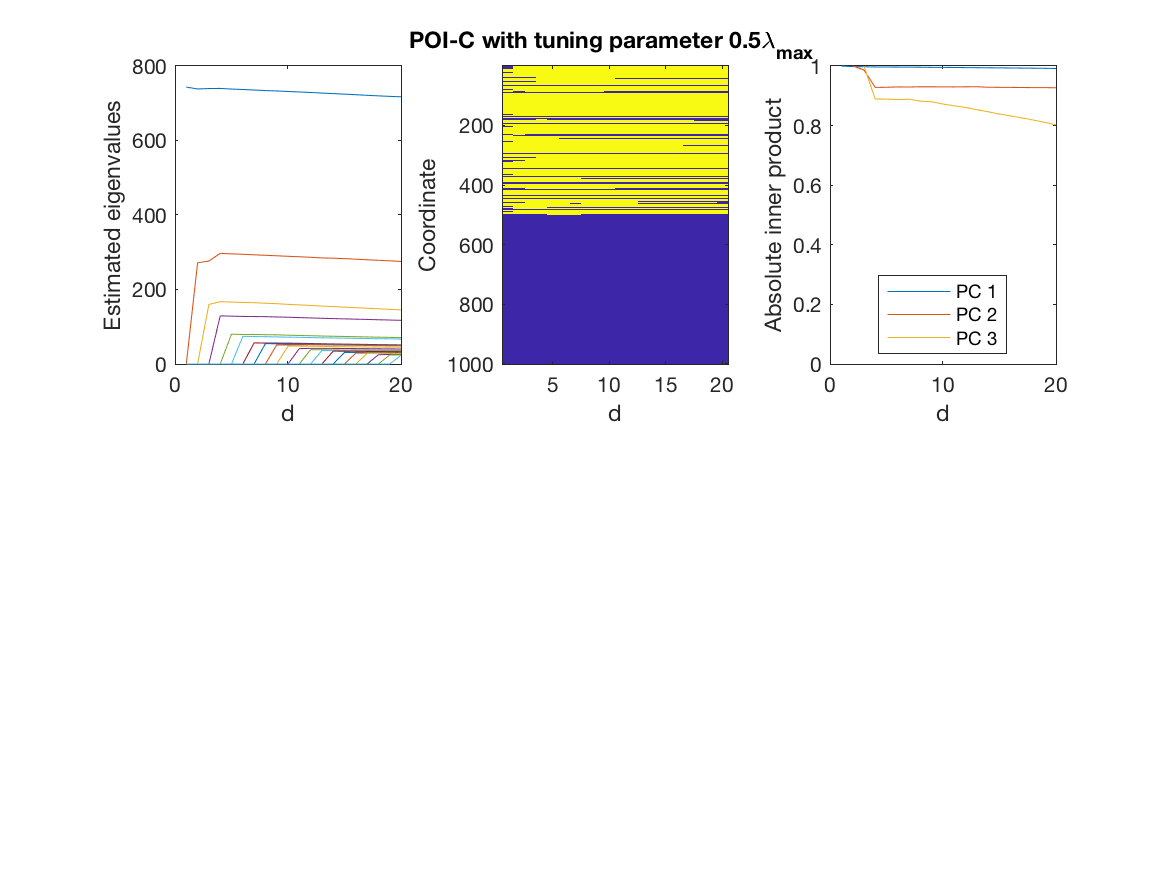}\\
  \vspace{-2.3in}
  \caption{Sparse PCA by POI with coordinate-wise sparse penalty for a genomic data set. The analysis is repeated for subspace dimension $d = 1$ to $20$. Shown are estimated eigenvalues (left), nonzero coordinates shown as lighter color (middle) and $|\hat{\qv}_{i,i}^\smt \hat{\qv}_{i,d}|$ for $i =1,2,3$, where $\hat{\qv}_{i,d}$ is the $i$th PC direction estimate when estimating PC subspace of dimension $d$ (right).
  }\label{fig:PCA_LF}

\end{figure}

We now present a simulation study to compare the performance of our application to sparse PCA with competing methods \citep{zou2006sparse,shen2008sparse,song2015sparse}.
We simulate from $p$-variate normal distribution with mean zero and covariance matrix $\boldsymbol{\Sigma}$. Let $d$ stand for the number of distinguishable principal components. We use $d = 3$ or $5$. Let $\boldsymbol{\Sigma} = \Uv_d \Lambdav_d \Uv_d^\smt + \Id_p$, where $\Lambdav_d$ is the diagonal matrix consisting of eigenvalues, satisfying $\mbox{diag}(\Lambdav_d^{1/2}) = 3 ( 5,4,\ldots,5-d+1)$.
We use three models for the $p \times d$ eigenvector matrix $\Uv_d$, defined below. 
\begin{itemize}
\item[Model I:] The eigenvector matrix $\Uv_d$ has only $s = 10$ nonzero rows. Specifically, the first $s$ elements in each column of $\Uv_d$ are $\zv / \norm{\zv}_2$, where $\zv$ is independently sampled from $N_s(\0v,\Id_s)$. The eigenvectors under this model are coordinate-wise sparse.
\item[Model II:] For $s = 5$, the first $ds$ rows of $\Uv_d$ are $s^{-1/2} \Id_d \otimes \1v_s$, and the rest of rows consist of zeros. Under this model, eigenvectors are not only orthogonal, but also combinations of disjoint sets of coordinates.
\item[Model III:] For $s = 5$, the eigenvector matrix $\Uv_d$ has $ds$ nonzero rows, formed similar to a block-lower-triangular matrix. Specifically, the nonzero rows are given by the QR decomposition of $\Av_d \otimes \1v_s$, where $\Av_d$ is a $ d \times d$ lower-triangular matrix, with all lower-triangular elements being one.
\end{itemize}

 Model I calls for a coordinate-wise sparse estimation, Model II for an element-wise sparse estimation, while Model III does not clearly favor any choice of penalty. Both Models II and III seem unnatural to be conceived as a model underlying any real data.  As argued in \cite{Bouveyron2016} coordinate-wise sparse principal component directions are more natural to interpret than element-wise sparse directions.
%

The estimates from our methods are denoted by POI-L, POI-C, FastPOI-L and FastPOI-C, where ``L'' stands for using the lasso penalty (\ref{eq:estimation1_LASSO}), and ``C'' stands for using the coordinate-wise group lasso penalty (\ref{eq:estimation1_GROUPLASSO}).
For each method, we define a candidate set of tuning parameters $L = \{ (0.75)^{i} \lambda_{\max}: i = 0, 1, \ldots, t, \infty \}$.
We set $t = 31$, and $\lambda_{\max}$ as the upper bound for $L$ as defined in  Remarks~\ref{remark2} and \ref{remark3}. 
For each choice of $\lambda \in L$, we compute the cross-validation score (\ref{eq:CV2}), using an independent tuning set of data.

The performance of each method is measured by a distance between the estimated subspace, spanned by $\widehat\Uv_d$, and the true eigenspace $\Uv_d$.
The principal angles $\theta_i$ between two subspaces are often used to measure the distance. The principal angles between $\widehat{\Uv} \in \Oc(p,d_1)$ and $\Uv \in \Oc(p,d_2)$ are defined as $\theta_i = \arccos( {\sigma}_i(\widehat{\Uv}^\smt\Uv))$, where ${\sigma}_i(\Av)$ is the $i$th largest singular value of $\Av$. We use the projection metric
\begin{equation}\label{eq:projectionmetric}
\rho(\widehat{\Uv}, \Uv) = \max\{\sin \theta_{i}:i=1,\ldots, \min(d_1,d_2) \}.
\end{equation}
If $d_1 = d_2 = d$, then the projection metric is equivalent to the difference between two corresponding projection matrices:
$\rho (\widehat{\Uv}, \Uv) = \|\widehat{\Uv}\widehat{\Uv}^\smt - \Uv\Uv^\smt \|,$ where $\|\cdot\|$ is the natural spectral norm.

Let $\widehat{\Uv}(\lambda)$ be the estimate of ${\Uv}$, when tuning parameter $\lambda$ is used.
Denote the cross-validated estimate by $\widehat{\Uv}(\hat\lambda)$ where  $\hat\lambda = {\rm arg} \max_{\lambda \in L} CV( \lambda)$.
The performances of subspace estimates are measured by the following criteria:
\begin{enumerate}
\item The minimal distance to the truth: $\min_{\lambda \in L} \rho(\widehat{\Uv}(\lambda), \Uv).$
\item The distance from the cross-validated estimate: $\rho(\widehat{\Uv}(\hat\lambda), \Uv).$
\end{enumerate}

\begin{table}[t]
{\footnotesize
\centering
\begin{tabular}{ccc|ccccccc}
Model & $d$ & $p$ & POI-L & POI-C & FastPOI-L & FastPOI-C & Zou et al. & Song et al. & Shen \& Huang \\ \hline
\hline
\multirow{4}{*}{I} &3&200&0.196&\textbf{0.159}&0.200&\textbf{0.162}&0.195&0.252&0.196\\
                   &3&500&0.197&\textbf{0.150}&0.210&0.156&0.197&0.276&0.194\\
                   &5&200&0.310&\textbf{0.196}&0.387&0.220&0.311&0.330&0.310\\
                   &5&500&0.348&\textbf{0.204}&0.450&0.363&0.303&0.448&0.343\\
\hline
\multirow{4}{*}{II} &3&200&\textbf{0.106}&0.162&0.155&0.160&0.124&\textbf{0.104}&\textbf{0.105}   \\
                    &3&500&\textbf{0.102}&0.164&0.155&0.164&0.123&\textbf{0.101}&\textbf{0.101}  \\
                    &5&200&\textbf{0.168}&0.496&0.332&0.458&0.242&0.228&0.215  \\
                    &5&500&\textbf{0.169}&0.561&0.397&0.699&0.290&0.308&0.207  \\
\hline
\multirow{4}{*}{III} &3&200&\textbf{0.108}&0.150&0.177&0.154&0.151&\textbf{0.105}&\textbf{0.106}\\
                     &3&500&\textbf{0.117}&0.154&0.195&0.161&0.164&\textbf{0.119}&\textbf{0.115} \\
                     &5&200&\textbf{0.214}&0.343&0.479&0.404&0.400&0.371&0.294  \\
                     &5&500&\textbf{0.222}&0.352&0.623&0.556&0.489&0.439&0.282
\end{tabular}}
\caption{The minimal projection distance to the truth, averaged from 100 repetitions, for principal subspace estimation.  The standard errors are at most 0.023. Smaller distance indicates more precise estimation. Highlighted are the best performed models (within 2 standard error of the smallest).
  \label{tab:simulation1}}
\end{table}

\begin{table}[t]
\centering
{\footnotesize
\begin{tabular}{ccc|ccccccc}
Model & $d$ & $p$ & POI-L & POI-C & FastPOI-L & FastPOI-C & Zou et al. & Song et al. & Shen \& Huang \\ \hline
\multirow{4}{*}{I } &3&200&0.202&\textbf{0.162}&0.204&\textbf{0.165}&0.201&0.262&0.198\\
                    &3&500&0.204&\textbf{0.152}&0.213&0.159&0.211&0.292&0.205\\
                    &5&200&0.359&\textbf{0.199}&0.420&0.228&0.419&0.445&0.369\\
                    &5&500&0.482&\textbf{0.209}&0.651&0.378&0.443&0.576&0.397\\
                   \hline
\multirow{4}{*}{II} &3&200&\textbf{0.111}&0.163&0.159&0.162&0.130&\textbf{0.107}&\textbf{0.109}\\
                    &3&500&\textbf{0.110}&0.166&0.158&0.166&0.127&\textbf{0.105}&\textbf{0.106}\\
                    &5&200&0.284&0.538&0.354&0.459&0.293&0.276&\textbf{0.232}\\
                    &5&500&0.376&0.620&0.630&0.705&0.331&0.361&\textbf{0.286}\\
                   \hline
\multirow{4}{*}{III}&3&200&\textbf{0.114}&0.151&0.180&0.155&0.161&\textbf{0.115}&0.135 \\
                    &3&500&\textbf{0.125}&0.156&0.197&0.162&0.169&\textbf{0.124}&0.146 \\
                    &5&200&0.420&\textbf{0.344}&0.509&0.407&0.530&0.502&0.405  \\
                    &5&500&0.537&\textbf{0.355}&0.741&0.558&0.554&0.588&0.435 \\
\end{tabular}}
\caption{The projection distance from the cross-validated estimate, averaged from 100 repetitions, for principal subspace estimation.  The standard errors are at most 0.027. Smaller distance indicates more precise estimation. Highlighted are the best performed models (within 2 standard error of the smallest).
  \label{tab:simulation2}}
\end{table}

For each model we choose the number of principal components as $d=3$ or 5, and choose the number of variables as $p =200$ or $500$. We use a small sample size of $n = 100$ for both training and tuning. The true number of principal components is treated as known. The empirical performances based on 100 repetitions of the experiments are summarized in Tables~\ref{tab:simulation1} and \ref{tab:simulation2}.
Table~\ref{tab:simulation1} summarizes the potential of each method, while Table~\ref{tab:simulation2} shows the actual numerical performance with the automatic tuning parameter selection. 

We note several observations from the simulation studies.
First, POI solutions are potentially closer to the truth than Fast POI solutions. 
Second, as expected, our methods with coordinate-wise sparsity-inducing penalty (POI-C and FastPOI-C) work well for the coordinate-wise sparse models (I and III), especially for the larger subspace dimension $d = 5$. In contrast, the lasso-type penalty works well for Model II, the eigenvectors of which are strictly element-wise sparse. 
Finally,  POI-C provides much more accurate estimates than the competing methods for coordinate-wise sparse models (Model I). 

In terms of the variable selection performance of the estimates, POI-C and FastPOI-C are both clearly superior than any other methods for Model I, and are comparable to the best performing methods for Models II and III. Numerical results for the variable selection performance can be found in the online supplementary material.

\subsection{Sparse Subspace Learning for Multiclass LDA}\label{sec:numerical_LDA}

We now apply our algorithm in the estimation of sparse discriminating basis for multi-group data.
Suppose that $\xv$ follows a $K$-mixture of multivariate normal distributions, each with mean $\muv_i$ and a common variance $\boldsymbol{\Sigma}_i = \boldsymbol{\Sigma}$, for all $i = 1, \ldots, K$, where $K>1$. Assuming that the group membership $y$ is also observed, write  $\boldsymbol{\Sigma}_T = \Cov(\xv)$ for the total-covariance matrix and $\boldsymbol{\Sigma}_W = \boldsymbol{\Sigma}$ for the within-group covariance matrix. The between-group covariance matrix is $\boldsymbol{\Sigma}_B = \boldsymbol{\Sigma}_T - \boldsymbol{\Sigma}_W$, whose rank is at most $K-1$. Then the problem of finding the discriminant subspace is equivalent to the GEP with
 $\Av = \boldsymbol{\Sigma}_B$ and $\Bv = \boldsymbol{\Sigma}_W$.

We demonstrate the performance of Fast POI, but not POI. There are three reasons for this. First of all, Fast POI is well-suited for the case where the rank of $\Av$ is known. In multiclass LDA, the rank of $\Av$ is typically $K-1$. Second, the performance of Fast POI was comparable or even superior to that of POI in our preliminary numerical studies. Third, POI requires considerably longer computation times than Fast POI. The performance of Fast POI estimates are compared with recently proposed linear sparse classifiers \citep{Clemmensen:2011,mai2015multiclass,gaynanova2016simultaneous}.


To simulate data, we use  mixtures of normal distributions $N(\muv_i, \boldsymbol{\Sigma})$.
We set the mean matrix $\muv = (\muv_1,\ldots, \muv_K)$ and $\boldsymbol{\Sigma}$ as follows. Throughout, the dimension is set to be $p = 200$.
Let $C(\rho)$ be the compound symmetry covariance model: $C(\rho) = (1-\rho) \Id_p + \rho \Jv_p$, $\Jv_p = \1v_p \1v_p^\smt$. Let $R(\rho)$ be the AR(1) model: $\{R(\rho)\}_{i,j} = \rho^{|i-j|}$.
Define $\Vv = (\vv_1,\vv_2,\vv_3)$ and $\Wv = (\wv_1,\wv_2,\wv_3)$ by
\begin{align*}
\vv_1^\smt &= (2,1,2,1,2, 0,\ldots,0)_{ 1\times p },    &      \wv_1^\smt &= (-1,1,1,1,1, 0,\ldots,0)_{ 1\times p },\\
\vv_2^\smt &= (1,-1,1,-1,1, 0,\ldots,0)_{ 1\times p },  &      \wv_2^\smt &= (1,-1,1,-1,1, 0,\ldots,0)_{ 1\times p },\\
\vv_3^\smt &= (0,1,-1,1,0, 0,\ldots,0)_{ 1\times p },   &      \wv_3^\smt &= (1,1,-1,1,0, 0,\ldots,0)_{ 1\times p }.
\end{align*}
Our models are given as follows. For Model I, $\muv = (\muv_1,\muv_2,\muv_3) = \Vv$, $\boldsymbol{\Sigma} = \Id_p$;
Model II, $\muv = \boldsymbol{\Sigma} \Vv$, $\boldsymbol{\Sigma} = C(0.5)$;
Model III, $\muv = \boldsymbol{\Sigma} \Vv$,  $\boldsymbol{\Sigma} = R(0.5)$;
Model IV,  $\muv = \boldsymbol{\Sigma} \Wv$, $\boldsymbol{\Sigma} = C(0.5)$; for
Model V, let $\tilde\Wv = 2(\wv_1,\wv_2,\wv_3, \bar{\wv})$, where  $\bar{\wv} = (\wv_1+\wv_2+\wv_3)/3$, and set
        $\muv = (\muv_1,\muv_2,\muv_3, \muv_4)  = \boldsymbol{\Sigma} \tilde\Wv$, and $\boldsymbol{\Sigma} = C(0.5)$.
        %
%
%
%
%
%

The true generalized eigenvectors are given by solving the generalized eigenvector problem with $\Av = \sum_{k = 1}^K (\muv_k - \bar\muv)(\muv_k - \bar\muv)^\smt / K$, where $\bar\muv = \sum_{k=1}^K \muv_k / K$, and $\Bv = \boldsymbol{\Sigma}$.  
Note that although the eigenvectors are different for different models, the true subspaces in Models I--III are the same; they are all spanned by $\{\vv_i - \vv_j, : i \neq j\}$. Likewise, the true subspaces in Models IV and V are both spanned by $\{\wv_i - \wv_j, : i \neq j\}$.
Since the basis of a subspace can only be coordinate-wise sparse, it is expected that coordinate-wise sparse estimates (including FastPOI-C and the method of \cite{gaynanova2016simultaneous}) work well. Our simulation result, reported below, concurs.

In Models I--IV, we need to estimate a subspace of dimension 2, since there are $K = 3$ groups. In Model V, we still need to estimate the 2-dimensional subspace, even though there are $K = 4$ groups. In principle, this information is not available to statisticians, and for all model estimation, we estimate the subspace of dimension $d = K-1$; that is, the model is misspecified for Model V. Our simulation shows that our method works well for this case as well.



Our application to sparse LDA as well as the competing methods \citep{Clemmensen:2011,mai2015multiclass,gaynanova2016simultaneous} first estimate sparse basis vectors, then apply multiclass LDA to the dimension-reduced data. Following their approaches, the standard LDA is applied to the data pair ($\Xv\widehat\Uv,y)$ for the classification and prediction. From each model above, we generate a training set of size $n_i = 30$ (for each group), a tuning set of the same size, and a testing set of size $n = 100Kn_i$. Our tuning procedure, utilizing (\ref{eq:CV2}), is used for all methods.
 %
 The testing set is used to estimate the misclassification rate. Experiments are repeated 100 times. We report the qualities of subspace learning and classification performance in Tables~\ref{tab:LDA-simulation-short} and \ref{tab:LDA-simulation2-short}, respectively.

Table \ref{tab:LDA-simulation-short} shows that FastPOI-C provides most accurate estimates of the the true subspace. As explained above, it makes most sense to assume coordinate-wise sparsity in subspace learning. Thus the methods of \cite{mai2015multiclass} and \cite{Clemmensen:2011}, seeking element-wise sparsity, are expected to be inferior in subspace learning. Likewise, the performance of FastPOI-L is inferior to that of FastPOI-C. Note that \cite{gaynanova2016simultaneous} also imposed the coordinate-wise sparsity,  thus showing better performance in subspace learning than other method, except FastPOI-C. 

The proposed FastPOI-C also performed the best in terms of classification. It is evident from Table \ref{tab:LDA-simulation2-short} that FastPOI-C yields the smallest misclassification rates for Models I--IV, and is second to Gaynanova's classifier for Model V. Given the similarity of FastPOI-C and Gaynanova's (note Remark \ref{remark:1-poi-Mai-Gaynanova} and that they both seek coordinate-wise sparse solutions), it is expected that both perform well in classification.

\begin{table}[t]
{
\centering
\begin{tabular}{c|ccccc}
Model  & FastPOI-L & FastPOI-C & Mai et al. & Clemmensen et al. & Gaynanova et al. \\ \hline
          I       &  \textbf{0.328} &   \textbf{0.313} &   0.371 &   0.622 &   0.377     \\
           II   &  0.839 &   \textbf{0.570} &   0.936 &   0.817 &   0.611    \\
            III   &  0.644 &   \textbf{0.437} &   0.542 &   0.784 &   0.608     \\
            IV    &  0.852 &   \textbf{0.478} &   0.916 &   0.689 &   0.514    \\
             V    &  0.712 &   \textbf{0.359} &   0.869 &   \textbf{0.365} &   0.411   \\
\end{tabular}}
\caption{The projection distance from the estimate, averaged from 100 repetitions, for sparse discriminant basis learning.  The standard errors are at most 0.024. Smaller distance indicates more precise estimation.
Highlighted are the best performed models (within 2 standard error of the smallest).
  \label{tab:LDA-simulation-short}}
\end{table}

\begin{table}[t]
\centering
\begin{tabular}{c|ccccc}
Model  & FastPOI-L & FastPOI-C & Mai et al. & Clemmensen et al. & Gaynanova et al. \\ \hline
                 I  &   \textbf{7.46}    &       \textbf{7.27}    &       9.41    &      10.23   &        12.30     \\
                  II   &  30.50   &        \textbf{8.72 }  &       21.74  &         9.70  &         12.49    \\
                    III    &   17.10    &      \textbf{12.41}    &      17.23  &        18.41    &      19.19    \\
                   IV    &   35.84    &    \textbf{16.03}    &     23.62  &       18.00    &     19.68      \\
                  V    &   32.40    &     \textbf{16.13}    &     26.80   &      \textbf{16.57}   &      \textbf{15.98}      \\
\end{tabular}
\caption{Misclassification rates (in percent) of the test set, averaged from 100 repetitions.  The standard errors are at most 1.19. Smaller error rate indicates better classification. Highlighted are the best performed models (within 2 standard error of the smallest).
  \label{tab:LDA-simulation2-short}}
\end{table}

An alternative tuning procedure, minimizing misclassification error rates in the special case of multiclass LDA, provides similar numerical performances. In fact, both choices of the tuning parameter are generally close to each other, as we have seen in Fig.~\ref{fig:tuning}. Fast POI tends to choose more variables than needed, but shows better sensitivity, i.e., more signal variables (or coordinates) are included in the estimates, than other methods.  In terms of an overall variable selection performance, FastPOI-C and \cite{gaynanova2016simultaneous}'s method shows similar performances, while both are better than the others. 
Simulation results for the alternative choices of tuning parameters and for the variable selection performance can be found in the online supplementary material.


\subsection{Sufficient Dimension Reduction}

%
%
%

Sufficient dimension reduction \citep{cook2009regression} aims to find a projection of data $\xv$ that is sufficient for (i.e., preserves all information about) the conditional distribution of $y$ given $\xv$. Many sufficient dimension reduction (SDR) methods can be cast into a GEP, as shown in \cite{li2007sparse}. In particular, the solution for the sliced inverse regression  \citep[SIR, ][]{li1991sliced}, the most well-known method for SDR, is given by the eigenvectors of the GEP (\ref{eq:gep0}) where $\Av = \mbox{Cov}[ \E\{\xv - \E(\xv) | y \}]$ and $\Bv = \mbox{Cov} (\xv)$. We apply our sparse solutions of GEP by POI-C to SIR for a benchmark data set. In all examples in this subsection, we used $\lambda = \lambda_{\max}/2$.

The Tai-Chi data have been sometimes used as a benchmark data for various methods of SDR, e.g., in \cite{wu2009sparse}.
For our purpose, we randomly generate $(\xv,y)$ pairs for $n = 1000$ times as shown in Fig.~\ref{fig:Tai-Chi}. The first two coordinates of $\xv \in \Real^p$ are uniformly sampled over the disk with radius 2,  while the rest of coordinates are sampled from the standard normal distribution. The binary variable $y$ depends on the location of $\xv$, and is 1 if $\xv$ lies on the `yin' part of the Tai-Chi symbol, 0 otherwise. The true subspace only spans over the first two coordinates.

For sufficient dimension reduction of the toy data, a direct competitor is \cite{chen2010coordinate}'s sparse estimation for SIR. Note that \cite{tan2016sparse}'s method is not applicable to this data set as it only computes one eigenvector, while at least two eigenvectors are needed here. The method of \cite{li2007sparse} is not compared, as \cite{chen2010coordinate} improves upon \cite{li2007sparse}. Figure~\ref{fig:Tai-Chi} displays the projected data onto the estimated subspaces from the original SIR and our adaptation of SIR with POI-C. The graphical result from \cite{chen2010coordinate} is almost identical to POI-C, and both methods greatly improve upon the original SIR for this data set.

To further compare our approach with  \cite{chen2010coordinate}'s, we have repeated the experiment for various cases of sample size and dimension.
It appears that  Chen's method requires large sample size to work well, and is numerically unstable for small sample sizes. Our method exhibits great performances for any case, and requires only a fraction of computation time compared to Chen's. Relevant numerical results can be found in the online supplementary material.

\begin{figure}[t]
  \centering
  \vskip -1in
  \includegraphics[width=0.8\textwidth]{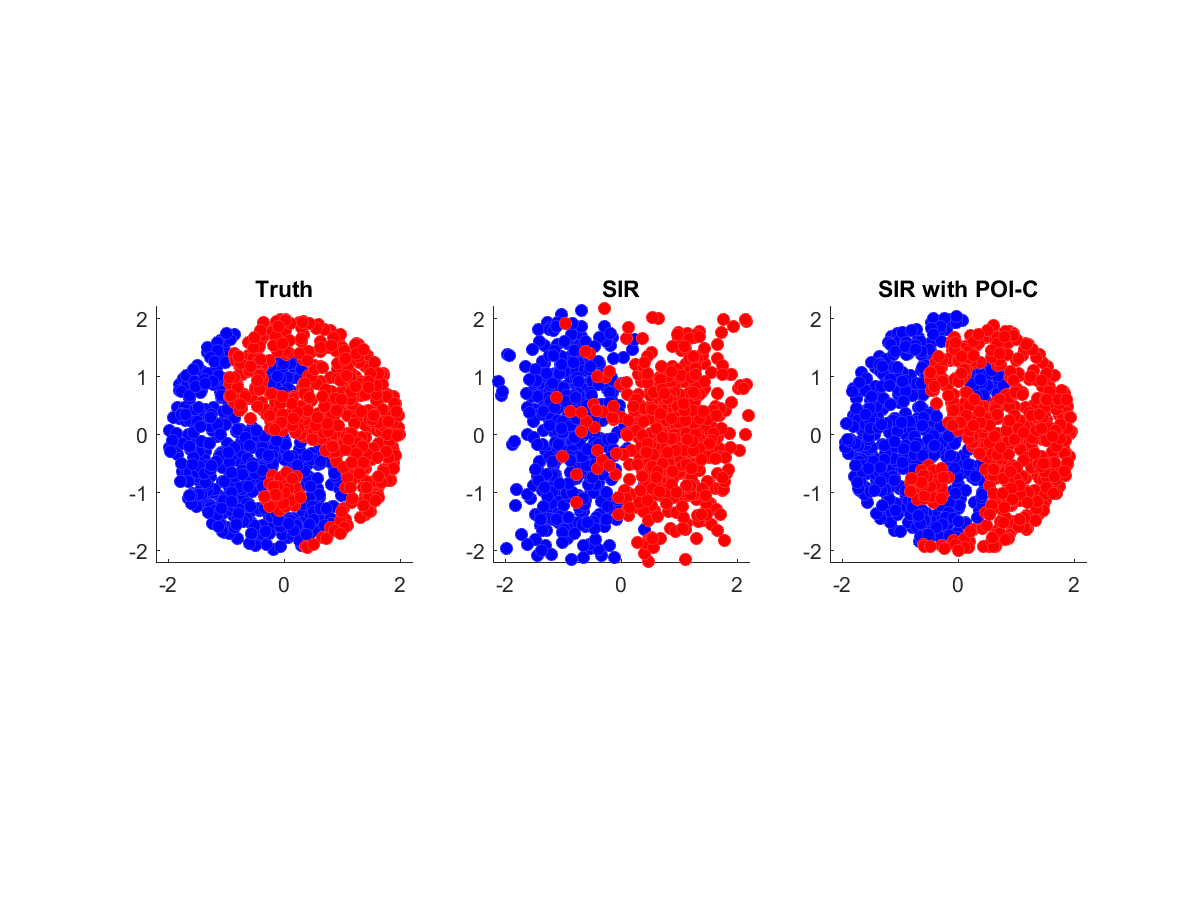}
  \vskip -1in
  \caption{Projections of a Tai-Chi data set to the true subspace (left panel), the estimated subspace by the original SIR \citep[][middle panel]{li1991sliced}, SIR estimated by POI-C (right panel). The data set has $n =1000$ observations, each with dimension $p = 500$.
  \label{fig:Tai-Chi}}
\end{figure}

\subsection{Canonical Correlation Analysis}

 We briefly demonstrate an application of our methods to sparse canonical correlation analysis (CCA). It is well known that CCA can be viewed as a GEP \citep{gaynanova2016penalized,Safo2018}. The canonical coefficient vectors are defined as follows. For two random vectors $\xv \in \Real^p$, $\yv \in \Real^q$, write  $\boldsymbol{\Sigma}_1 = \Cov(\xv)$, $\boldsymbol{\Sigma}_2 = \Cov(\yv) $ and $\boldsymbol{\Sigma}_{12} = \Cov(\xv, \yv) = \boldsymbol{\Sigma}_{21}^\smt$. The coefficient vectors for the first pair of canonical variables are $(\gv_1, \hv_1) \in \Real^p \times \Real^q$,
 maximizing the correlation between $\gv^\smt \xv$ and $\hv^\smt \xv$:
$\rho (\gv,\hv) = {\gv^\smt \boldsymbol{\Sigma}_{12} \hv} / {(\gv^\smt\boldsymbol{\Sigma}_1\gv)^\half (\hv^\smt\boldsymbol{\Sigma}_2 \hv)^{\half}}.
$
Since $\rho(\gv,\hv)$ is invariant under individual scaling of $(\gv,\hv)$, a Lagrangian formulation of the maximization involves the condition
$\gv^\smt\boldsymbol{\Sigma}_1 \gv = \hv^\smt\boldsymbol{\Sigma}_2\hv = 1$, and the first-order condition of the Lagrangian coincides with the GEP (\ref{eq:gep0}) for
\begin{equation}\label{eq:CCA-1}
\Av = \left(\begin{array}{cc}
                \0v & \boldsymbol{\Sigma}_{12} \\
                 \boldsymbol{\Sigma}_{21} &\0v
 \end{array}\right), \quad \Bv = \left(\begin{array}{cc}
                \boldsymbol{\Sigma}_1 & \0v \\
                \0v & \boldsymbol{\Sigma}_2
 \end{array}\right),
\end{equation}
where the solution $(\uv, \lambda)$ corresponds to the concatenated coefficient vector $(\gv^\smt, \hv^\smt)^\smt$ and canonical correlation coefficient $\rho(\gv, \hv)$, respectively.
An alternative formulation of CCA is given by solving the GEP (\ref{eq:CCA-1}) with respect to individual $\gv$ or $\hv$, leading to two GEPs:
 \begin{align}\label{eq:CCA-2}
\Sigmav_{12}\Sigmav_{2}^{-1}\Sigmav_{21} \gv = \lambda \Sigmav_{1}\gv, \quad
\Sigmav_{21}\Sigmav_{1}^{-1}\Sigmav_{12} \hv = \lambda \Sigmav_{2}\hv.
\end{align}

 In performing a sparse CCA, we use (\ref{eq:CCA-2}) as done in \cite{Safo2018}.
In the estimation of the pair $(\gv,\hv)$ from a sample, we follow the suggestion from \cite{witten2009penalized} and \cite{Safo2018} of first standardizing the data, then replacing $\widehat\Sigmav_1$ by $\Iv_p$ and $\widehat\Sigma_2$ by $\Iv_q$.
With the sample cross-covariance matrix $\widehat\Sigmav_{12}$,
we apply the POI in solving
$\Sigmav_{12}\Sigmav_{12}^\smt \gv = \lambda \gv$ and
$\Sigmav_{12}^\smt\Sigmav_{12} \hv = \lambda \hv$. Applying our method, or any existing general sparse GEP methods \citep{sriperumbudur2011majorization,song2015sparse}, to the GEP (\ref{eq:CCA-1}) typically leads unsatisfactory results partly due to the fact that $\Av$ is not in general nonnegative definite and also because of different scales of $\Sigmav_{1}$ and $\Sigmav_{2}$. In a simulation study with $n =80$ observations of two random vectors of dimension $p = 200, q = 150$, in which there is one canonical pair, the performance of POI estimates is comparable to the method of \cite{Safo2018}, and are superior than the methods of \cite{witten2009penalized} and \cite{gao2017sparse}. We refer to the online supplementary material for simulation settings and numerical results.

\appendix

\section{Technical Details}\label{sec:technicalLemma}

\begin{proof}[Proof of Proposition~\ref{thm1}]
 There exists a nonsingular matrix $\Sv$ such that $\Uv_d = \tilde\Qv_d \Sv$. The equation  (\ref{eq:gepU}) is written as
\begin{equation*}
\Av \tilde\Qv_d \Sv = \Bv \tilde\Qv_d \Sv \Lambdav_d,
\end{equation*}
which yields
\begin{equation*}
(\tilde\Qv_d^\smt\Av \tilde\Qv_d) \Sv = (\tilde\Qv_d^\smt\Bv \tilde\Qv_d) \Sv \Lambdav_d,
\end{equation*}
moreover, we have $\Sv^\smt(\tilde\Qv_d^\smt\Bv \tilde\Qv_d) \Sv = \Uv_d^\smt\Bv\Uv_d = \Iv_d$.
Therefore, $\Sv$ and $\Lambdav_d$ respectively are the matrix of eigenvectors and the diagonal matrix of eigenvalues of the GEP problem (\ref{eq:gepdual}).
\end{proof}



The following elementary lemma is used in the proof of Proposition~\ref{lem:optionDsubspaceequality}. We use the notation $\Cc(\Av)$ to denote the column space of matrix $\Av$.

\begin{lem}\label{lem:simplelinearalgebra}
Let $\Mv, \Mv_1,\Mv_2 \in \Real^{p\times d} , \Nv \in \Real^{d\times d}, \Lv \in \Real^{p\times p}$ and assume that both $\Nv, \Lv$ are invertible.
\begin{enumerate}
\item $\Cc(\Mv) = \Cc(\Mv\Nv)$.
\item If $\Cc(\Mv_1) = \Cc(\Mv_2)$, then $\Cc(\Lv\Mv_1) = \Cc(\Lv\Mv_2)$.
\end{enumerate}
\end{lem}
We omit the proof of Lemma~\ref{lem:simplelinearalgebra}.

\begin{proof}[Proof of Proposition~\ref{lem:optionDsubspaceequality}]
Under Case (a), $\Bv = \Iv_p$, and the generalized eigen-equation becomes $\Av\Qv_d = \Qv_d \Lambdav_d^*$. Thus $\Qv_d = \Vv_d$ by definition, which in turn leads to
   $\Cc(\Bv^{-1}\Vv_d) = \Cc(\Vv_d) = \Cc(\Qv_d)$.

  For case (b), define $\Cv = \Bv^{-\half}\Av\Bv^{-\half}$. Then the generalized eigen-equation rewrites to
  $\Cv\Bv^{\half} \Qv_d = \Bv^{\half} \Qv_d \Lambdav_d^*.$
  By the QR decomposition of $\Bv^{\half}\Qv_d$, there exist $\Qv_d^\dagger \in \Oc(p,d)$ and a $k \times k$ upper-triangular matrix $\Lambdav^\dagger_d$ satisfying
  \begin{equation}
  \Cv \Qv_d^\dagger = \Qv_d^\dagger  \Lambdav^\dagger_d. \label{eq:daggereq}
  \end{equation}
  Since $\rank(\Cv) = k$, (\ref{eq:daggereq}) is equivalent to
  \begin{equation}
  \Cv [\Qv_d^\dagger, \Qv_d^{\perp}] = [\Qv_d^\dagger, \Qv_d^{\perp}] \begin{pmatrix}
  \Lambdav^\dagger_d & \0v \\
  \0v  & \0v \\
\end{pmatrix},  \label{eq:daggereq2}
  \end{equation}
  where $[\Qv_d^\dagger, \Qv_d^{\perp}] \in \Oc(p)$. Thus $\Cv = \Qv_d^\dagger  \Lambdav^\dagger_d (\Qv_d^\dagger)^\smt$, and  we get $\Cc(\Cv) = \Cc(\Qv_d^\dagger) = \Cc(\Bv^{\half}\Qv_d)$. (These are obtained by the definition of eigendecomposition and QR decomposition.)
  On the other hand, using Lemma~\ref{lem:simplelinearalgebra}, it can be shown that
  $\Cc(\Bv^{\half} \Bv^{-1}\Vv_d) =  \Cc( \Bv^{-\half}\Vv_d) = \Cc(\Bv^{-\half}\Av \Bv^{-\half} ) = \Cc(\Cv)$.

  Since $\Bv^\half$ is invertible and $\Cc(\Bv^{\half} \Bv^{-1}\Vv_d) =  \Cc(\Bv^{\half}\Qv_d)$, again by Lemma~\ref{lem:simplelinearalgebra}, we conclude that
  $\Cc(\Bv^{-1}\Vv_d) =  \Cc(\Qv_d)$.
\end{proof}

\begin{proof}[Proof of Lemma~\ref{lem:closeformsolutionforEvalues}]
Notice that each optimization problem is column-wise separable. Each column-wise subproblem is then a least-square problem.
\end{proof}

\section*{Acknowledgements}
Jeon's research was supported by Basic Science Research Program of the National Research Foundation of Korea (NRF-2015R1A1A1A05001180) funded by the Korean government.

\bigskip
\begin{center}
{\large\bf SUPPLEMENTARY MATERIAL}
\end{center}

\begin{description}

\item[Additional analysis:] Document containing an extensive list of statistical GEP problems and additional numerical results. (.pdf file)

\item[Matlab routine:] Matlab functions that perform the sparse generalized eigenvector estimation as described in the article. (Compressed files .zip ) The files are also available at \texttt{https://github.com/sungkyujung/SparseEIG}

\end{description}

\bibliographystyle{Chicago}
\bibliography{SpCDPbib}

\newpage
\includepdf[pages=-]{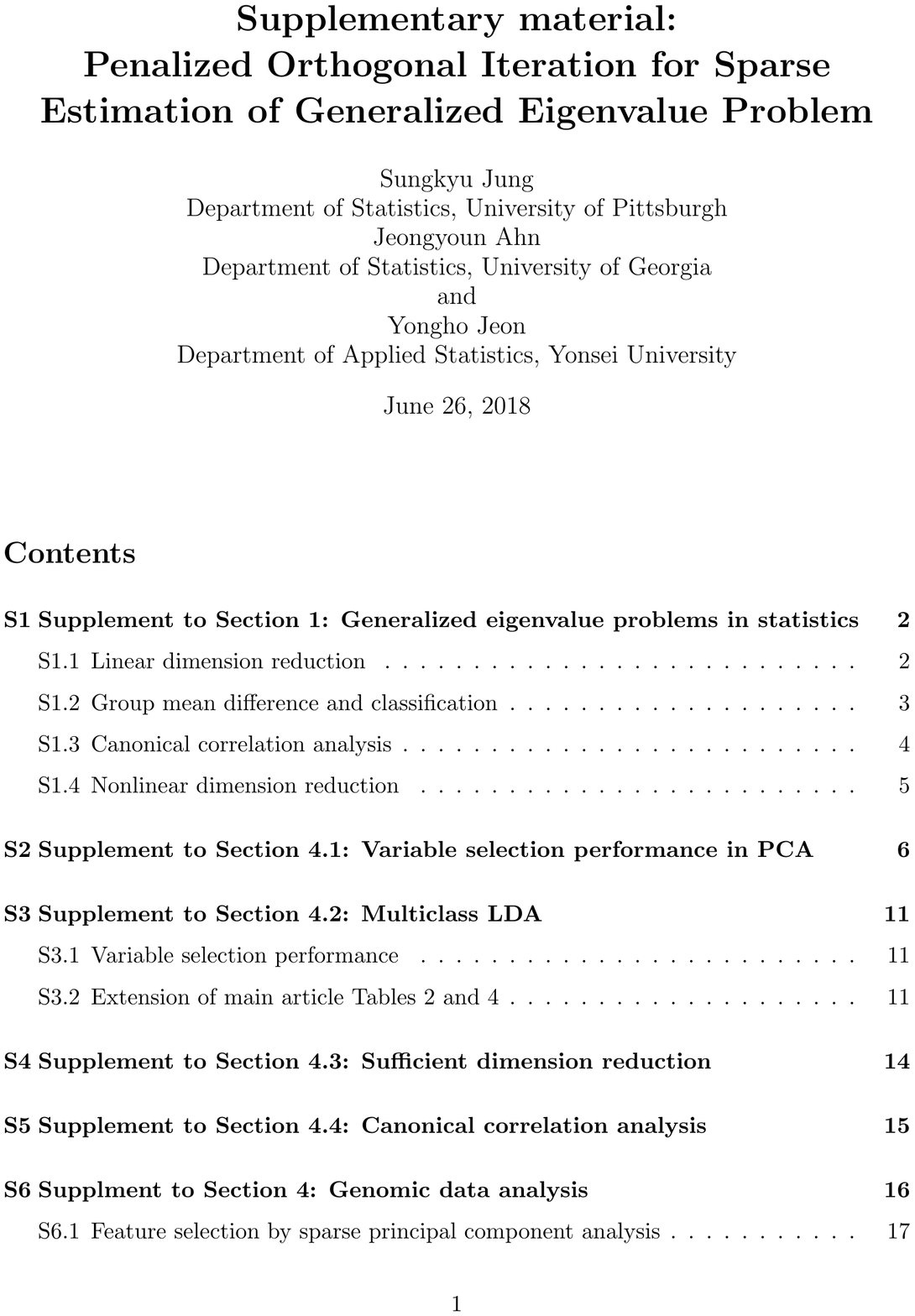}

\end{document}